\documentclass[11pt]{article}
\usepackage{amsmath}
\usepackage{amsmath,amsfonts,amssymb,graphicx,makeidx,amsthm}
\usepackage{amscd,amsfonts,amssymb,multicol}
\usepackage{amssymb, amsmath}
\usepackage{amsfonts}
\usepackage{amssymb}
\usepackage{graphicx}
\providecommand{\U}[1]{\protect\rule{.1in}{.1in}}
\newtheorem{theorem}{Theorem}[section]

\newtheorem{lemma}[theorem]{Lemma}
\newtheorem{proposition}[theorem]{Proposition}
\theoremstyle{definition}

\theoremstyle{remark}
\newtheorem{example}[theorem]{Example}
\newtheorem{remark}[theorem]{Remark}
\numberwithin{equation}{section}

\begin{document}

\title{A transmutation operator method for solving the inverse quantum scattering problem}
\author{Vladislav V. Kravchenko$^{1}$, Elina L. Shishkina$^{2}$ and Sergii M.
Torba$^{1}$\\{\small $^{1}$ Departamento de Matem\'{a}ticas, Cinvestav, Unidad
Quer\'{e}taro, }\\{\small Libramiento Norponiente \#2000, Fracc. Real de Juriquilla,
Quer\'{e}taro, Qro., 76230 MEXICO.}\\{\small $^{2}$ Voronezh State University.}\\{\small e-mail: vkravchenko@math.cinvestav.edu.mx,
storba@math.cinvestav.edu.mx, \thanks{Research was supported by CONACYT,
Mexico via the projects 222478 and 284470. Research of Vladislav Kravchenko
was supported by the Regional mathematical center of the Southern Federal
University, Russia.}}}
\maketitle

\begin{abstract}
The inverse quantum scattering problem for the perturbed Bessel equation is
considered. A direct and practical method for solving the problem is proposed.
It allows one to reduce the inverse problem to a system of linear algebraic
equations, and the potential is recovered from the first component of the
solution vector of the system. The approach is based on a special form
Fourier-Jacobi series representation for the transmutation operator kernel and
the Gelfand-Levitan equation which serves for obtaining the system of linear
algebraic equations. The convergence and stability of the method are proved as
well as the existence and uniqueness of the solution of the truncated system.
Numerical realization of the method is discussed. Results of numerical tests
are provided revealing a remarkable accuracy and stability of the method.

\end{abstract}

\section{Introduction}

We present a direct and simple method for practical solution of the inverse
quantum scattering problem for the perturbed Bessel equation
\[
Lu:=-u^{\prime\prime}+\left(  \frac{\ell(\ell+1)}{x^{2}}+q(x)\right)
u=\rho^{2}u,\quad x>0
\]
with an arbitrary fixed angular momentum $\ell\geq-1/2$ and the potential $q$
satisfying
\begin{equation}
\label{Cond on q}\int_{0}^{\infty}(x^{\mu}+ x) |\tilde q(x)|\,dx < +\infty
\end{equation}
for some $\mu\in[0,1/2)$, where
\begin{equation}
\label{tilde q}\tilde q(x) =
\begin{cases}
q(x), & \ell>-1/2,\\
\bigl(1+|\log(x)|\bigr)q(x), & \ell=-1/2.
\end{cases}
\end{equation}
The problem consists in recovering $q$ from the given scattering data. The
bibliography dedicated to the theory of this problem and applications is vast.
We refer to \cite{AgrMarch}, \cite{Chadan}, \cite{Coz 1976},
\cite{LevitanInverse}, \cite{Teschl} and references therein. However, the
numerical solution of the problem presents difficulties. We refer the reader
to \cite{Airapetyan 1997} and \cite{Kukulin 2004} where numerical approaches
are discussed, although the way of presenting the numerical results does not
give us a possibility to make a comparison.

The method presented in this paper allows one a direct reduction of the
inverse quantum scattering problem to a system of linear algebraic equations.
Moreover, only the first component of the solution vector is necessary to
recover the potential. The method is simple and does not require much
programmer's work. It is based on the classical results from the spectral
theory, such as the Gelfand-Levitan equation and the transmutation operator,
as well as on a new functional series representation for the transmutation
integral kernel, obtained in \cite{KT2020improvedNeumann}.

The present work extends the applicability of the approach based on the
functional series representations for the transmutation integral kernels
developed in the regular case $\ell=0$ in \cite{Kr2019JIIP}, \cite{Kr2019MMAS
InverseAxis}, \cite{DKK2019MMAS}, \cite{KKK QuantumFest} and reported in the
book \cite{KrBook2020}. The extension of an approach onto the singular case
$\ell\neq0$ is always a challenge requiring additional ideas and tools. The
first important\ ingredient here is an appropriate Fourier-Jacobi series
representation for the transmutation operator kernel
\cite{KT2020improvedNeumann}. It captures singular features of the kernel,
such as its behaviour near $t=0$ and on the characteristic line $t=x$, and
allows one to recover the potential from the first coefficient of the series.
Thus, we do not follow the usual approach of computing the transmutation
kernel first and then recovering the potential from it. Instead, we compute
the first coefficient of the Fourier-Jacobi series representation, from which
the potential is recovered.

The right choice of the orthogonal function system used in the series
representation resulted to be of crucial importance in the interplay between
the transmutation operator kernel and the Gelfand-Levitan input kernel, which
gave us the possibility in the present work to obtain a system of linear
algebraic equations for the coefficients of the series representation with
explicit formulas for the entries of the system matrix.

We prove the convergence and stability of the method. This results in the
possibility of recovering the potential from noisy scattering data. A
corresponding numerical example is provided. Moreover, we prove the existence
and uniqueness of the solution of the truncated system of equations arising in
the numerical realization of the method.

Thus, the method developed in the present work is convergent, stable and
possesses an important additional advantage. Its numerical implementation is
simple and does not require much programmer effort. The numerical examples
reveal a remarkable accuracy, stability and fast convergence of the method.

Besides this introduction the paper contains four sections. Section 2 presents
some preliminaries on the inverse quantum scattering problem including the
example of the square well potential, which is used later on for one of the
numerical tests. In Section 3 the Fourier-Jacobi series representation for the
transmutation operator kernel is presented. It is explained how the potential
can be recovered from the first coefficient of the series, and the
Gelfand-Levitan equation is recalled. In Section 4 we construct the system of
linear algebraic equations for the coefficients of the Fourier-Jacobi series
representation, prove the existence and uniqueness of solutions of
corresponding truncated systems and the convergence of solutions of truncated
systems to the exact one. Observing that the obtained truncated systems result
from applying the Bubnov-Galerkin procedure with a special choice of the
orthogonal function system, we prove the stability of the method, which allows
one to work efficiently with noisy scattering data. In Section 5 we discuss
the numerical implementation of the method and provide some numerical
examples. They illustrate that indeed the developed approach leads to a direct
and simple method for accurate recovering of the potential even with few
equations in the truncated system and from noisy scattering data. Finally, in
Appendix \ref{SectAppendix} we present a refined asymptotics of the Jost function.

\section{Preliminaries}

\label{Sect2}

We consider the perturbed Bessel equation%
\begin{equation}
Lu:=-u^{\prime\prime}+\left(  \frac{\ell(\ell+1)}{x^{2}}+q(x)\right)
u=\rho^{2}u,\quad x>0 \label{perturbed Bessel equation}%
\end{equation}
with the coefficient $q$, often called the potential, being a real valued
function satisfying the condition
\begin{equation}
(x^{\mu}+x)\tilde q(x) \in L_{1}(0,\infty)\qquad\text{for some }0\leq\mu<1/2,
\label{Condition on q}%
\end{equation}
Here $\tilde q$ is given by \eqref{tilde q}. Sometimes, potentials satisfying
(\ref{Condition on q}) at infinity are said to belong to the Marchenko class.
The spectral parameter $\rho\in\mathbb{C}$ is chosen so that
$\operatorname{Im}\rho\geq0$ and $\ell\geq-1/2$.

We are interested in a procedure for solving the inverse quantum scattering
problem consisting in recovering a potential $q(x)$ in the perturbed Bessel
equation from so-called scattering data which include the eigenvalues, the
corresponding norming constants and the Jost function $F_{\ell}(\rho)$,
$\rho\in[0,\infty)$. Notice that we suppose the Jost function to be given,
although in a usual study of the inverse problem it is obtained first from the
$S$-function (the scattering function) which is supposed in its turn to be
known as a part of the scattering data.

The unique solvability of such inverse quantum scattering problem follows from
\cite[Theorem 5.1]{Teschl}, where a more general class of potentials is
considered for arbitrary $\ell\geq-1/2$. Additional restrictions on the
potential imposed in this paper are needed to guarantee that the problem
possesses at most a finite number of eigenvalues, to use the Gelfand-Levitan
equation and to be sure that the solution of the Gelfand-Levitan equation is
square-integrable. For the case of integer $\ell$ one can consult a lot of
additional details, e.g., in \cite{Sta2} and \cite{Chadan}.

We remind that the set of eigenvalues, if it is not empty, consists of a
finite set of numbers $\rho_{j}^{2}\leq0$, $j=1,\ldots,N$, which are such that
equation (\ref{perturbed Bessel equation}) admits a square integrable solution
on $\left(  0,\infty\right)  $, see \cite[(II.1.10a)]{Chadan}, \cite[Theorem
5.1]{Seto1974} and \cite[Section 9.7]{Teschl2014}. Thus, $\rho_{j}=i\tau_{j}$,
$\tau_{j}\geq0$. For recalling the definition of the norming constants and of
the Jost function we proceed with some necessary notations.

A solution $\varphi_{\ell}(\rho,x)$ of (\ref{perturbed Bessel equation})
satisfying the asymptotic relation at the origin
\[
\lim_{x\rightarrow0}\frac{2^{\ell+1}}{\sqrt{\pi}}\Gamma\left(  \ell+\frac
{3}{2}\right)  x^{-(\ell+1)}\varphi_{\ell}(\rho,x)=1,
\]
is called the regular solution. Note that for integer values of $\ell$ one has
$\frac{2^{\ell+1}}{\sqrt{\pi}}\Gamma\left(  \ell+\frac{3}{2}\right)
=(2\ell+1)!!$. The last formula is known as the extension of the double
factorial symbol to complex arguments. To simplify notations, later in this
paper we will use $(2\ell+1)!!$.

In the case when $\rho=\rho_{j}$ is an eigenvalue, the regular solution
$\varphi_{\ell}(\rho_{j},x)$ is an eigenfunction, and the norming constants
are defined as
\[
c_{j}:=\frac{1}{\int_{0}^{\infty}\varphi_{\ell}^{2}(\rho_{j},x)dx}.
\]

A solution $f_{\ell}(\rho,x)$ of (\ref{perturbed Bessel equation}) satisfying
the asymptotic relation at infinity%
\[
\lim_{x\rightarrow\infty}\left(  e^{-\frac{i\pi\ell}{2}}e^{-i\rho x}f_{\ell
}(\rho,x)\right)  =1
\]
is called the Jost solution. The uniqueness and the existence of both regular
and Jost solutions is a well known fact (see, e.g., \cite{Chadan}, and for
non-integer values of $\ell$, \cite{Teschl2016}, \cite{HKT2018} and references therein).

The function $F_{\ell}(\rho)$ which can be represented as a Wronskian of the
solutions%
\[
F_{\ell}(\rho)=(-\rho)^{\ell}W\left[  f_{\ell}(\rho,x),\varphi_{\ell}%
(\rho,x)\right]
\]
is known as the Jost function. In fact, the Jost function contains information
on the behaviour of the Jost solution at the origin. The following asymptotic
relation is valid for $\ell>-1/2$ (see, e.g., \cite[Section 1.5]{Chadan}, for
non-integer values of $\ell$ it can be established using the results from
\cite{Teschl2016})
\begin{equation}
F_{\ell}(\rho)=\lim_{x\rightarrow0}\frac{\left(  -\rho x\right)  ^{\ell}%
}{\left(  2\ell-1\right)  !!}f_{\ell}(\rho,x),\label{F_l}%
\end{equation}
while for $\ell=-1/2$ it can be deduced (see \cite[Subsection 2.1]{HKT2018})
that
\begin{equation}
F_{-1/2}(\rho)=\lim_{x\rightarrow0}-\frac{\sqrt\pi(-\rho)^{-1/2}}{\sqrt
{2x}\log x}f_{-1/2}(\rho,x).\label{F_lcritical}%
\end{equation}

Note that $F_{\ell}$ is analytic in the upper half-plane, $F_{\ell}%
(\rho)=1+o(1)$ when $\rho\rightarrow\infty$, $\operatorname{Im}\rho\geq0$, and
$F_{\ell}(-\rho)=\overline{F_{\ell}}(\rho)$ for $\rho\in\mathbb{R}$
\cite[Lemma B.5]{Teschl}. Moreover, for $\ell>-1/2$ and potentials $q$ such
that $q\in L_{1}(0,\infty)$ the asymptotic formula is valid%
\[
F_{\ell}(\rho)=1+\frac{i}{2\rho}\int_{0}^{\infty}q(x)\,dx+o(\rho^{-1}%
),\qquad|\rho|\rightarrow\infty,\label{F asymptoticsKT}%
\]
see \cite[Remark 2.14]{Teschl2016}. In Appendix \ref{SectAppendix} we prove a
refinement of this formula, namely, that
\begin{equation}
F_{\ell}(\rho)=1+\frac{i}{2\rho}\int_{0}^{\infty}q(x)\,dx+O(\rho^{-2}%
),\qquad|\rho|\rightarrow\infty,\label{F asymptotics}%
\end{equation}
for any $\ell\ge-1/2$ and potentials $q\in L_{1}(0,\infty)\cap BV_{0}%
[0,\infty)$. Here $BV_{0}$ denotes functions of bounded variation vanishing at infinity.

Denote by $b_{\ell}(\rho x)$ a solution of the Bessel equation
\[
-u^{\prime\prime}+\frac{\ell(\ell+1)}{x^{2}}u=\rho^{2}u,\qquad x>0
\]
satisfying the asymptotic relation
\begin{equation}
b_{\ell}(\rho x)\sim\frac{\left(  \rho x\right)  ^{\ell+1}}{\left(
2\ell+1\right)  !!},\qquad x\rightarrow0. \label{asymptotics b_l}%
\end{equation}
It has the form
\[
b_{\ell}(\rho x)=\rho xj_{\ell}(\rho x)
\]
where $j_{\ell}$ stands for the spherical Bessel function of the first kind
(see \cite[Section 10.1]{AbramowitzStegunSpF}), $j_{\ell}(z):=\sqrt{\frac{\pi
}{2z}}J_{\ell+\frac{1}{2}}(z)$.

\begin{example}
\label{Ex01} Consider the square well potential $q$ of the form%
\begin{equation}
q(x)=
\begin{cases}
-Q^{2}, & x\leq R,\\
0, & x>R
\end{cases}
\label{square well}%
\end{equation}
where $Q$ is a positive constant. Denote $\omega:=\sqrt{\rho^{2}+Q^{2}}$. Then
the Jost solution has the form%
\[
f_{\ell}(\rho,x)=
\begin{cases}
a( \rho) b_{\ell}(\omega x)+b( \rho) \omega xh_{\ell}^{( 1) }(\omega x), &
x\leq R,\\
(-1)^{\ell}i\rho xh_{\ell}^{( 1) }(\rho x), & x>R
\end{cases}
\]
where the coefficients $a( \rho) $ and $b( \rho) $ are found from the
condition of continuity of the solution $f_{\ell}(\rho,x)$ and of its
derivative at $x=R$, which leads to the following system of equations%
\begin{align*}
a( \rho) b_{\ell}(\omega R)+b( \rho) \omega Rh_{\ell}^{( 1) }(\omega R)  &
=(-1)^{\ell}i\rho Rh_{\ell}^{( 1)}(\rho R),\\
a( \rho) \left(  \omega(\ell+1)j_{\ell}(\omega R)-\omega^{2}Rj_{\ell+1}(\omega
R)\right)   &  +b( \rho) \left(  \omega(\ell+1)h_{\ell}^{( 1) }(\omega
R)-\omega^{2}Rh_{\ell+1}^{( 1) }(\omega R)\right) \\
&  =(-1)^{\ell}i\rho\left(  (\ell+1)h_{\ell}^{( 1) }(\rho R)-\rho Rh_{\ell
+1}^{( 1) }(\rho R)\right)  .
\end{align*}
From \eqref{F_l} we find that%
\[
F_{\ell}(\rho)=(-1)^{\ell+1}ib( \rho) \left(  \frac{\rho}{\omega}\right)
^{\ell}.
\]

\end{example}

\section{The transmutation integral kernel}

A solution $u_{\ell}(\rho,x)$ of (\ref{perturbed Bessel equation}), satisfying
the asymptotic relation (\ref{asymptotics b_l}) admits the following
representation%
\[
u_{\ell}(\rho,x)=T\left[  b_{\ell}(\rho x)\right]  :=b_{\ell}(\rho x)+\int
_{0}^{x}K_{\ell}(x,t)b_{\ell}(\rho t)dt
\]
where the integral kernel $K_{\ell}(x,t)$ is a square integrable function of
the variable $t$, independent of $\rho$. This Volterra integral operator of
the second kind is known as a transmutation (or transformation) operator. The
existence of such $K_{\ell}(x,t)$ for the potentials satisfying condition
\eqref{Cond on q} at zero was proved in \cite{Sta2}. Properties of $K_{\ell
}(x,t)$ were studied in several publications (see, e.g., \cite{Volk},
\cite{Coz 1976}, \cite{Chebli}, \cite{Holz2020}, \cite{KatrakhovSitnik},
\cite{KrShishkinaTorba2018}, \cite{SitnikShishkina}, \cite{SitnikShishkina
Elsevier}). For the purpose of the present work the following statement is crucial.

\begin{theorem}
[\cite{KT2020improvedNeumann}]\label{Thm principal} Let $q$ satisfy the
condition $\int_{0}^{b} x^{\mu}|q(x)|\,dx<\infty$ for some $0\le\mu<1/2$. Then
the kernel $K_{\ell}(x,t)$ admits the following series representation%
\begin{equation}
K_{\ell}(x,t)=\sum_{n=0}^{\infty}\frac{\beta_{n}(x)}{x^{\ell+2}}t^{\ell
+1}P_{n}^{(\ell+1/2,0)}\left(  1-\frac{2t^{2}}{x^{2}}\right)  ,
\label{K series}%
\end{equation}
where $P_{n}^{(\alpha,\beta)}$ stands for the Jacobi polynomial and the
coefficients $\beta_{n}(x)$ can be calculated by a recurrent integration
procedure, starting with
\begin{equation}
\beta_{0}(x)=\left(  2\ell+3\right)  \left(  \frac{u_{\ell,0}(x)}{x^{\ell+1}%
}-1\right)  \label{beta_0}%
\end{equation}
where $u_{\ell,0}(x)$ is a regular solution of the equation
\begin{equation}
Lu=0 \label{Lu=0}%
\end{equation}
normalized by the asymptotic condition $u_{\ell,0}(x)\sim x^{\ell+1}$,
$x\rightarrow0$.

For any $x>0$, the series in \eqref{K series} converges in $L_{2}(0,x)$.
Suppose additionally that $q$ is absolutely continuous on $[0,b]$. Then the
series in \eqref{K series} converges absolutely and uniformly with respect to
$t\in\left[  0,x-\varepsilon\right]  $ for an arbitrarily small $\varepsilon
>0$. If additionally $q\in W_{1}^{3}[0,b]$, then the series converges
absolutely and uniformly with respect to $t$ on the whole $[0,x]$.
\end{theorem}

\begin{remark}
The condition on the potential $q$ in the theorem is equivalent to condition
\eqref{Cond on q} at the origin. The recurrent integration procedure mentioned
in the theorem is superfluous for the present work and can be consulted in
\cite{KT2020improvedNeumann}.
\end{remark}

\begin{remark}
Equality \eqref{beta_0} gives us the possibility to recover the potential $q$
if $\beta_{0}$ is known. Indeed, we have that
\begin{equation}
u_{\ell,0}(x)=\left(  \frac{\beta_{0}(x)}{\left(  2\ell+3\right)  }+1\right)
x^{\ell+1}, \label{ul0}%
\end{equation}
and since $u_{\ell,0}$ is a solution of \eqref{Lu=0}, we obtain%
\begin{equation}
q=\frac{x\beta_{0}^{\prime\prime}(x)+2(\ell+1)\beta_{0}^{\prime}(x)}{x\left(
\beta_{0}(x)+2\ell+3\right)  }. \label{q=beta}%
\end{equation}

\end{remark}

\begin{remark}
\label{Rem Jacobi orthonormal}The following orthogonality property of the
Jacobi polynomials is valid \cite{KT2020improvedNeumann}%
\begin{equation}
\int_{0}^{x}t^{2\ell+2}P_{n}^{(\ell+1/2,0)}\left(  1-\frac{2t^{2}}{x^{2}%
}\right)  P_{m}^{(\ell+1/2,0)}\left(  1-\frac{2t^{2}}{x^{2}}\right)
dt=\frac{x^{2\ell+3}}{4m+2\ell+3}\delta_{nm} \label{orthogonality Jacobi}%
\end{equation}
with $\delta_{nm}$ standing for the Kronecker delta. Consequently, for any
$x>0$ fixed, the system of functions
\begin{equation}
p_{n}(x;t):=\frac{\sqrt{4n+2\ell+3}}{x^{\ell+3/2}}t^{\ell+1}P_{n}%
^{(\ell+1/2,0)}\left(  1-\frac{2t^{2}}{x^{2}}\right)  \label{pn orthonormal}%
\end{equation}
is a complete orthonormal system in $L_{2}(0,x)$. Hence the series
\eqref{K series} is an expansion of the kernel $K_{\ell}(x,t)$ with respect to
the basis of $L_{2}(0,x)$ represented by the system of functions $\left\{
p_{n}(x;t)\right\}  _{n=0}^{\infty}$,%
\begin{equation}
K_{\ell}(x,t)=\sum_{n=0}^{\infty}\alpha_{n}(x)p_{n}(x;t) \label{Kn expansion}%
\end{equation}
with
\[
\alpha_{n}(x)=\frac{\beta_{n}(x)}{\sqrt{4n+2\ell+3}\sqrt{x}}.
\]

\end{remark}

In the following we assume that zero is not an eigenvalue of the problem. Then
the transmutation kernel $K_{\ell}(x,t)$ is related to the scattering data via
the Gelfand-Levitan integral equation
\begin{equation}
K_{\ell}(x,y)+\Omega_{\ell}(x,y)+\int_{0}^{x}K_{\ell}(x,t)\Omega_{\ell
}(t,y)dt=0,\qquad x>y\label{GL equation}%
\end{equation}
where the input kernel $\Omega_{\ell}(x,y)$ has the form%
\begin{gather*}
\Omega_{\ell}(x,y)=\sum_{j=1}^{N}C_{j}b_{\ell}(i\tau_{j}x)b_{\ell}(i\tau
_{j}y)+\frac{2}{\pi}\int_{0}^{\infty}b_{\ell}(\rho x)b_{\ell}(\rho y)\left(
\left\vert F_{\ell}(\rho)\right\vert ^{-2}-1\right)  d\rho,\\
C_{j}:=\frac{c_{j}}{\left(  i\tau_{j}\right)  ^{2\ell+2}}.
\end{gather*}

Note that under condition \eqref{Cond on q} the integral kernel $K_{\ell}$
satisfies \cite{Sta2} for any finite $a>0$
\[
\sup_{0\leq x\leq a}\Vert K_{\ell}(x,\cdot)\Vert_{L_{2}(0,x)}^{2}<\infty.
\]
The function $\Omega_{\ell}$ is symmetric, and it can be easily obtained from
\eqref{GL equation} that $\Omega_{\ell}(x,\cdot)\in L_{2}(0,x)$ and
$\Omega_{\ell}\in L_{2}((0,x)\times(0,x))$.

\section{A system of linear algebraic equations for the coefficients
$\beta_{n}(x)$}

Denote%
\begin{equation}
A_{m,n}(x):=\sum_{j=1}^{N}C_{j}j_{\ell+2n+1}(i\tau_{j}x)j_{\ell+2m+1}%
(i\tau_{j}x)+\frac{2}{\pi}\int_{0}^{\infty}j_{\ell+2n+1}(\rho x)j_{\ell
+2m+1}(\rho x)\left(  \left\vert F_{\ell}(\rho)\right\vert ^{-2}-1\right)
d\rho, \label{Amn}%
\end{equation}
and%
\begin{equation}
B_{m}(x):=-\sum_{j=1}^{N}C_{j}b_{\ell}(i\tau_{j}x)j_{\ell+2m+1}(i\tau
_{j}x)-\frac{2}{\pi}\int_{0}^{\infty}b_{\ell}(\rho x)j_{\ell+2m+1}(\rho
x)\left(  \left\vert F_{\ell}(\rho)\right\vert ^{-2}-1\right)  d\rho.
\label{Bm}%
\end{equation}

\begin{theorem}
\label{Th system for coefficients}The coefficients $\beta_{n}$ from
\eqref{K series} satisfy the following infinite system of linear algebraic
equations%
\begin{equation}
\frac{\beta_{m}(x)}{\left(  4m+2\ell+3\right)  x}+\sum_{n=0}^{\infty}\beta
_{n}(x)A_{m,n}(x)=B_{m}(x),\qquad\text{for all }m=0,1,\ldots. \label{system}%
\end{equation}

\end{theorem}

\begin{proof} Let us substitute the representation (\ref{K series}) into
(\ref{GL equation}). Consider
\begin{equation}
\int_{0}^{x}K_{\ell}(x,t)\Omega_{\ell}(t,y)dt=\frac{1}{x^{\ell+2}}\sum_{n=0}^{\infty
}\beta_{n}(x)\int_{0}^{x}t^{\ell+1}P_{n}^{(\ell+1/2,0)}\left(  1-\frac{2t^{2}%
}{x^{2}}\right)  \Omega_{\ell}(t,y)dt. \label{vsp1}%
\end{equation}
The possibility of changing the order of summation and integration follows
from the observation that this equality is nothing but a concrete realization
of the general Parseval identity \cite[p. 16]{AkhiezerGlazman}. Indeed, with
the aid of Remark \ref{Rem Jacobi orthonormal} we have
\begin{align*}
\int_{0}^{x}K_{\ell}(x,t)\Omega_{\ell}(t,y)dt  &  =\left\langle K_{\ell}(x,\cdot
),\Omega_{\ell}(\cdot,y)\right\rangle _{L_{2}(0,x)}\\
\displaybreak[2]
&  =\sum_{n=0}^{\infty}\left\langle K_{\ell}(x,\cdot),p_{n}(x;\cdot)\right\rangle
_{L_{2}(0,x)}\left\langle p_{n}(x;\cdot),\Omega_{\ell}(\cdot,y)\right\rangle
_{L_{2}(0,x)}\\
\displaybreak[2]
&  =\sum_{n=0}^{\infty}\alpha_{n}(x)\left\langle p_{n}(x;\cdot),\Omega
_{\ell}(\cdot,y)\right\rangle _{L_{2}(0,x)}\\
\displaybreak[2]
&  =\sum_{n=0}^{\infty}\frac{\beta_{n}(x)}{\sqrt{4n+2\ell+3}\sqrt{x}}\int_{0}%
^{x}\frac{\sqrt{4n+2\ell+3}}{x^{\ell+3/2}}t^{\ell+1}P_{n}^{(\ell+1/2,0)}\left(
1-\frac{2t^{2}}{x^{2}}\right)  \Omega_{\ell}(t,y)dt\\
&  =\frac{1}{x^{\ell+2}}\sum_{n=0}^{\infty}\beta_{n}(x)\int_{0}^{x}t^{\ell+1}%
P_{n}^{(\ell+1/2,0)}\left(  1-\frac{2t^{2}}{x^{2}}\right)  \Omega_{\ell}(t,y)dt.
\end{align*}
In order to proceed with the integral in (\ref{vsp1}), we need the following
result \cite{KT2020improvedNeumann}%
\begin{equation}
\int_{0}^{x}t^{\ell+1}P_{n}^{(\ell+1/2,0)}\left(  1-\frac{2t^{2}}{x^{2}}\right)
b_{\ell}(\rho t)dt=x^{\ell+2}j_{\ell+2n+1}(\rho x). \label{integral b_l}%
\end{equation}
Hence
\begin{multline}\label{Omega expansion}
\int_{0}^{x}t^{\ell+1}P_{n}^{(\ell+1/2,0)}\left(  1-\frac{2t^{2}}{x^{2}%
}\right)  \Omega_{\ell}(t,y)dt\\
=x^{\ell+2}\biggl(  \sum_{j=1}^{N}C_{j}j_{\ell+2n+1}(i\tau_{j}x)b_{\ell}(i\tau
_{j}y)+\frac{2}{\pi}\int_{0}^{\infty}j_{\ell+2n+1}(\rho x)b_{\ell}(\rho y)\left(
\left\vert F_{\ell}(\rho)\right\vert ^{-2}-1\right)  d\rho\biggr)  ,
\end{multline}
and
\begin{multline*}
\int_{0}^{x}K_{\ell}(x,t)\Omega_{\ell}(t,y)dt\\
=\sum_{n=0}^{\infty}\beta_{n}(x)\biggl(  \sum_{j=1}^{N}C_{j}j_{\ell+2n+1}%
(i\tau_{j}x)b_{l}(i\tau_{j}y)+\frac{2}{\pi}\int_{0}^{\infty}j_{\ell+2n+1}(\rho
x)b_{\ell}(\rho y)\left(  \left\vert F_{\ell}(\rho)\right\vert ^{-2}-1\right)
d\rho\biggr)  .
\end{multline*}
Thus, equation (\ref{GL equation}) can be written in the form%
\begin{multline}
\frac{y^{\ell+1}}{x^{\ell+2}}\sum_{n=0}^{\infty}\beta_{n}(x)P_{n}%
^{(\ell+1/2,0)}\left(  1-\frac{2y^{2}}{x^{2}}\right) \\
+\sum_{n=0}^{\infty}\beta_{n}(x)\biggl(  \sum_{j=1}^{N}C_{j}j_{\ell+2n+1}%
(i\tau_{j}x)b_{\ell}(i\tau_{j}y)+\frac{2}{\pi}\int_{0}^{\infty}j_{\ell+2n+1}(\rho
x)b_{\ell}(\rho y)\left(  \left\vert F_{\ell}(\rho)\right\vert ^{-2}-1\right)
d\rho\biggr) \\
=-\sum_{j=1}^{N}C_{j}b_{\ell}(i\tau_{j}x)b_{\ell}(i\tau_{j}y)-\frac{2}{\pi}%
\int_{0}^{\infty}b_{\ell}(\rho x)b_{\ell}(\rho y)\left(  \left\vert F_{\ell}%
(\rho)\right\vert ^{-2}-1\right)  d\rho\label{GL series}%
\end{multline}
Multiplying (\ref{GL series}) by $y^{\ell+1}P_{m}^{(\ell+1/2,0)}\left(
1-\frac{2y^{2}}{x^{2}}\right)  $, integrating with respect to $y$ from $0$ to
$x$, and using (\ref{integral b_l}) and (\ref{orthogonality Jacobi}) we obtain
(\ref{system}). The series in (\ref{system}) converges again due to the
general Parseval identity because it is a scalar product of the functions
$K_{\ell}(x,t)$ and $\int_{0}^{x}\Omega_{\ell}(t,y)p_{m}(x;t)dt$ in the space
$L_{2}(0,x)$.
\end{proof}

The functions $\frac{\beta_{m}(x)}{\sqrt{4m+2\ell+3}\sqrt x}$ are the Fourier
coefficients of the function $K_{\ell}(x,\cdot)$ with respect to the system
\eqref{pn orthonormal}, see \eqref{Kn expansion}. It follows from
\eqref{Omega expansion} that the functions $\sqrt{4m+2\ell+3}\sqrt x\cdot
B_{m}(x)$ are the Fourier coefficients of the function $-\Omega_{\ell}%
(x,\cdot)$ with respect to the system \eqref{pn orthonormal}. Finally,
multiplying \eqref{Omega expansion} by $y^{\ell+1}P_{m}^{(\ell+1/2,0)}\left(
1-\frac{2y^{2}}{x^{2}}\right)  $, integrating with respect to $y$ from $0$ to
$x$ and using (\ref{integral b_l}) we obtain that
\[
x^{2\ell+4}A_{n,m}(x) = \int_{0}^{x} \int_{0}^{x} t^{\ell+1} P_{n}%
^{(\ell+1/2,0)}\left(  1-\frac{2t^{2}}{x^{2}}\right)  y^{\ell+1} P_{m}%
^{(\ell+1/2,0)}\left(  1-\frac{2y^{2}}{x^{2}}\right)  \Omega_{\ell
}(t,y)\,dt\,dy,
\]
or that
\[
\sqrt{4n+2\ell+3}\sqrt{4m+2\ell+3}\cdot x A_{m,n}(x) = \int_{0}^{x} \int
_{0}^{x} p_{n}(t) p_{m}(y) \Omega_{\ell}(t,y)\,dt\,dy.
\]
The last equality means that the functions $\sqrt{4n+2\ell+3}\sqrt{4m+2\ell
+3}\cdot x A_{m,n}(x)$ are the Fourier coefficients of the function
$\Omega_{\ell}$ with respect to the system $p_{n}\times p_{m}$.

Hence for each fixed $x>0$ the infinite system \eqref{system} can be written
as
\begin{equation}
\xi_{j}-\lambda\sum_{k=0}^{\infty}a_{jk}\xi_{k}=b_{j},\qquad j=0,1,\ldots,
\label{system Kantorovich}%
\end{equation}
where $\lambda=-1$ and
\[
\xi_{j}=\frac{\beta_{j}(x)}{\sqrt{4j+2\ell+3}\sqrt{x}},\quad b_{j}%
=\sqrt{4j+2\ell+3}\sqrt{x}\cdot B_{j}(x),\quad a_{jk}=\sqrt{4j+2\ell+3}%
\sqrt{4k+2\ell+3}\cdot xA_{j,k}(x).
\]
The coefficient vectors satisfy $\{b_{j}\}_{j=0}^{\infty}\in\ell_{2}$,
$\{a_{j,k}\}_{j,k=0}^{\infty}\in\ell_{2}\otimes\ell_{2}$ and the unknown
vector $\{\xi_{j}\}_{j=0}^{\infty}$ is sought to belong to $\ell_{2}$. The
systems of such type, with coefficients from $\ell_{2}$, were studied in
\cite[Chapter 14, \S 3]{KantorovichAkilov}, and the following result follows immediately.

\begin{proposition}
\label{Prop Convergence} Let $x>0$ be fixed. Consider the system
\eqref{system} truncated to $M+1$ equations, i.e., we consider $m,n\le M$.
Then for sufficiently large $M$ the truncated system has a unique solution
which we denote by $\{\beta_{m}^{(M)}(x)\}_{m=0}^{M}$ and
\[
\sum_{m=0}^{M} \frac{|\beta_{m}(x) - \beta_{m}^{(M)}(x)|^{2}}{(4m+2\ell+3)x} +
\sum_{m=M+1}^{\infty}\frac{|\beta_{m}(x)|^{2}}{(4m+2\ell+3)x}\to0,\qquad
M\to\infty,
\]
from which it also follows that
\[
\beta_{0}^{(M)}(x)\to\beta_{0}(x),\qquad M\to\infty.
\]

\end{proposition}

The same truncated system results from the application of the Bubnov-Galerkin
procedure to the integral equation \eqref{GL equation} with respect to the
system \eqref{pn orthonormal}, see \cite[\S 14]{Mihlin}. However, in our
approach we do not need to solve the complete system, only the first function
$\beta_{0}$ is necessary to recover the potential. Also we point out that the
special form of the function system \eqref{pn orthonormal} allowed us to
transform the scalar products arising in the Bubnov-Galerkin procedure into
the form \eqref{Amn} and \eqref{Bm}. As a consequence of the general theory
presented in \cite[\S 14]{Mihlin} we obtain a stability result for the
proposed method.

Let $I_{M}$ be the $(M+1)\times(M+1)$ identity matrix, $L_{M}=(a_{jk}%
)_{j,k=0}^{M}$ be the coefficient matrix of the truncated system and
$R_{M}=(b_{j})_{j=0}^{M}$ be the truncated right-hand side. Following
\cite[\S 9]{Mihlin} consider a system (called non-exact system)
\[
(I_{M}+L_{M}+\Gamma_{M})v=R_{M}+\delta_{M},
\]
where $\Gamma_{M}$ is an $(M+1)\times(M+1)$ matrix representing errors in the
coefficients $a_{jk}$, and $\delta_{M}$ is a column-vector representing errors
in the coefficients $b_{j}$. Let $U_{M}$ denote the solution of the exact
truncated system (with $\Gamma_{M}=0$ and $\delta_{M}=0$) and $V_{M}$ the
solution of the non-exact system. Note that $U_{M}=\left\{ \frac{\beta^{M}%
_{m}(x)}{\sqrt{4m+2\ell+3}\sqrt{x}}\right\} _{m=0}^{M}$, see Proposition
\ref{Prop Convergence}. The solution of the Bubnov-Galerkin procedure is
called stable if there exist constants $c_{1}$, $c_{2}$ and $r$ independent of
$M$ such that for $\Vert\Gamma_{M}\Vert\leq r$ and arbitrary $\delta_{M}$ the
non-exact system is solvable and the following inequality holds
\[
\Vert U_{M}-V_{M}\Vert\leq c_{1}\Vert\Gamma_{M}\Vert+c_{2}\Vert\delta_{M}%
\Vert.
\]
From \cite[Theorems 14.1 and 14.2]{Mihlin} the following result follows.

\begin{proposition}
\label{Prop CondNum} The approximate solution $\left\{ \frac{\beta^{M}_{m}%
(x)}{\sqrt{4m+2\ell+3}\sqrt{x}}\right\} _{m=0}^{M}$ of system
\eqref{system Kantorovich} is stable. Moreover, the condition numbers of the
coefficient matrices $I_{M}+L_{M}$ are bounded.
\end{proposition}

This result allows one to recover the potential from inexact to a certain
point or noisy scattering data.

\section{Numerical implementation}

\subsection{General scheme}

Theorem \ref{Th system for coefficients} and Proposition
\ref{Prop Convergence} lead to a direct and simple method for solving the
inverse quantum scattering problem.

\begin{enumerate}
\item Given the Jost function, the eigenvalues and the norming constants.
Choose a number of equations $M+1$, so that the truncated system
\begin{equation}
\frac{\beta_{m}(x)}{\left(  4m+2\ell+3\right)  x}+\sum_{n=0}^{M}\beta
_{n}(x)A_{m,n}(x)=B_{m}(x),\qquad\text{for all }m=0,\ldots,M
\label{truncated system A}%
\end{equation}
is to be solved.

\item Compute $B_{m}(x)$ and $A_{m,n}(x)$ according to the formulas (\ref{Bm})
and (\ref{Amn}).

\item Solve the system (\ref{truncated system A}) to find $\beta_{0}(x)$.

\item Compute $q$ with the aid of (\ref{q=beta}) or by computing first the
particular solution $u_{\ell,0}$ using (\ref{ul0}).
\end{enumerate}

\begin{remark}
Since the condition numbers of truncated systems \eqref{system Kantorovich}
are bounded, see Proposition \ref{Prop CondNum}, it may be worth converting
the system \eqref{truncated system A} into the truncated system of the form
\eqref{system Kantorovich} for large values of $M$.
\end{remark}

\subsection{On calculation of the integrals}

\label{SubsectAsymptotics} Calculation of the integrals in \eqref{Amn} and
\eqref{Bm} is one of the key steps in the proposed method. Since for $\rho
\in\mathbb{R}$
\[
|j_{\nu}(\rho x)|=\frac{\cos(\rho x-\frac{1}{2}\pi\nu-\frac{\pi}{4})}{|\rho
x|}+O\left(  \frac{1}{|\rho x|^{2}}\right)  ,\qquad|\rho|\rightarrow\infty,
\]
see \cite[(9.2.1)]{AbramowitzStegunSpF} and
\begin{equation}
\left\vert F_{\ell}(\rho)\right\vert ^{-2}-1=\frac{1}{4\rho^{2}}\left(
\int_{0}^{\infty}q(x)\,dx\right)  ^{2}+O\left(  \frac{1}{\rho^{2}}\right)
=O\left(  \frac{1}{\rho^{2}}\right) ,\qquad|\rho|\rightarrow\infty,
\label{F asymptotics2}%
\end{equation}
see \eqref{F asymptotics}, we have
\[
\left\vert j_{\ell+2n+1}(\rho x)j_{\ell+2m+1}(\rho x)\left(  \left\vert
F_{\ell}(\rho)\right\vert ^{-2}-1\right)  \right\vert \leq\frac{c_{1}}%
{x^{2}\rho^{4}}%
\]
and
\[
\left\vert b_{\ell}(\rho x)j_{\ell+2m+1}(\rho x)\left(  \left\vert F_{\ell
}(\rho)\right\vert ^{-2}-1\right)  \right\vert \leq\frac{c_{2}}{x\rho^{3}%
},\qquad\rho\rightarrow+\infty.
\]
As one can see, the integral in \eqref{Bm} can converge slowly. The
convergence can be improved to some extent subtracting leading term and
integrating it explicitly. Note that due to \eqref{F asymptotics2},
\begin{equation}
\rho^{2}\left(  \left\vert F_{\ell}(\rho)\right\vert ^{-2}-1\right)
=O(1),\qquad|\rho|\rightarrow\infty, \label{MainTermFl}%
\end{equation}
that is, a bounded term. Numerical experiments suggest that this bounded term
is a sum of a constant, an oscillating function and an $o(1)$ function. The
value of the constant, which we will denote by $\tilde{F}_{\ell}$, can be
easily estimated numerically. For example, one can compute the expression
\eqref{MainTermFl} for some set of points $\{\rho_{k}\}_{k=0}^{K}$ and take
for $\tilde{F}_{\ell}$ an average of the obtained values. See Figure
\ref{Fig Fl} for an illustration.

Note also that (see \cite[2.12.31.2]{Prudnikov})
\begin{equation}
\int_{0}^{\infty}\frac{j_{\ell+2n+1}(\rho x)j_{\ell+2m+1}(\rho x)}{\rho^{2}%
}\,d\rho=%
\begin{cases}
\frac{\pi x}{4(\ell+2n+1/2)_{3}}, & \text{if }n=m,\\
\frac{\pi x}{8(\ell+n+m+1/2)_{3}}, & \text{if }n=m\pm1,\\
0, & \text{if }|n-m|\geq2,
\end{cases}
\label{Int1}%
\end{equation}
and
\begin{equation}
\int_{0}^{\infty}\frac{b_{\ell}(\rho x)j_{\ell+2m+1}(\rho x)}{\rho^{2}}%
\,d\rho=%
\begin{cases}
\frac{\pi x}{2(\ell+1/2)_{2}}, & \text{if }m=0,\\
0, & \text{if }m>0,
\end{cases}
\label{Int2}%
\end{equation}
where $(x)_{n}$ stands for the Pochhammer symbol. Hence instead of computing
integrals \eqref{Amn} and \eqref{Bm} one can compute the integrals
\begin{equation}
\int_{0}^{\infty}j_{\ell+2n+1}(\rho x)j_{\ell+2m+1}(\rho x)\left(  \left\vert
F_{\ell}(\rho)\right\vert ^{-2}-1-\frac{\tilde{F}_{\ell}}{\rho^{2}}\right)
d\rho\label{AmnMod}%
\end{equation}
and
\begin{equation}
\int_{0}^{\infty}b_{\ell}(\rho x)j_{\ell+2m+1}(\rho x)\left(  \left\vert
F_{\ell}(\rho)\right\vert ^{-2}-1-\frac{\tilde{F}_{\ell}}{\rho^{2}}\right)
d\rho\label{BmMod}%
\end{equation}
and afterwards add expressions \eqref{Int1} and \eqref{Int2} multiplied by
$\tilde{F}_{\ell}$. If the integrals are truncated and computed on a segment
$[0,K]$, the proposed modification leads to a more accurate result due to the
integral tail taken into account (the oscillating and $o(1)$ parts in
\eqref{MainTermFl} are expected to result in smaller values in comparison with
the part given by the constant $\tilde F_{\ell}$). We would like to mention
that the proposed modification is nothing more than an adaptation of the
method presented in \cite[(9.98)]{Piessens2000} with first two terms taken
into account. Note also that in the case $\ell=-1/2$ the expression
\eqref{Int1} can not be used for $n=m=0$, and the expression \eqref{Int2} for
$m=0$ (due to the divergence at the origin). One should use directly
expressions \eqref{Amn} and \eqref{Bm}. For $\ell<0$ the integrals
\eqref{AmnMod} and \eqref{BmMod} possess integrable singularity at the origin.

In the present work we opted out of applying special methods for calculating
oscillatory integrals and the Hankel transform in particular (see, e.g.,
\cite{Piessens2000}, \cite{Zaman2019} and references therein). The main reason
is that the function $F_{\ell}$ oscillates a lot even for simplest potentials,
see Figure \ref{Fig Fl}. For that reason we are not expecting a simple
approximation of the term $\left\vert F_{\ell}(\rho)\right\vert ^{-2}-1$ to be
possible and decided to leave a detailed study for a separate work.

\begin{figure}[tbh]
\centering
\includegraphics[bb=0 0 360 144, width=5in,height=2in]{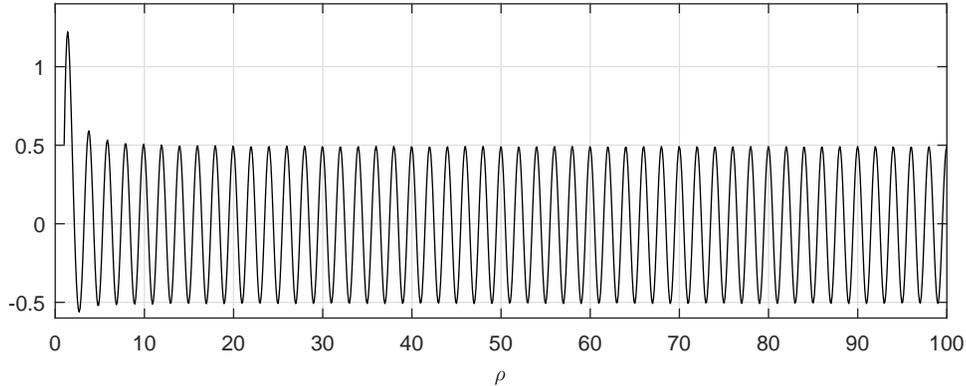}
\caption{Plot of the function $\rho^{2} \left(  \left|  F_{\ell}(\rho)\right|  ^{-2}-1-\frac{\tilde F_{\ell}}{\rho^{2}}\right)  $ for the square well potential from Example \ref{Ex01} with $\ell=1$, $Q=1$ and $R=\pi/2$. The parameter
$\tilde F_{\ell}$ is estimated numerically to be $1.5079$.}
\label{Fig Fl}
\end{figure}

\subsection{Numerical examples}

The numerical illustrations presented below were obtained in Matlab2017. For
the numerical integration on step 2 a sufficiently large interval was chosen
and the Matlab routine \texttt{trapz} was used. On the last step, for
recovering $q$ we used (\ref{q=beta}). Here the differentiation was performed
by representing the computed function $\beta_{0}(x)$ in the form of a spline
with the aid of the Matlab routine \texttt{spapi} with a posterior
differentiation with the Matlab command \texttt{fnder}.

\begin{example}
\label{Ex 1}Consider the potential \eqref{square well} with $\ell=2$, $Q=1$
and $R=\pi/2$. On Figure \ref{Fig 1} the recovered potential (on the left) and
its absolute error (on the right) are shown in the cases $M=0$, $M=1$, $M=4$
and $M=9$ that corresponds to 1, 2, 5 and 10 equations in the truncated system
\eqref{truncated system A}, respectively. Thus, a very reduced number of
equations from the system \eqref{truncated system A} is sufficient even in the
case of a discontinuous potential. For the numerical integration we have used
the interval $\rho\in\lbrack0,5000]$. However it should be mentioned that such
a large interval was used only to demonstrate that the method can recover
smooth potentials with a 10 decimal digits accuracy. Reducing the integration
to the interval $[0,100]$ and taking the step-size of $1/10$ for the
\texttt{trapz} function still allowed us to recover the potential with 4--5
decimal figures.

\begin{figure}[tbh]
\centering
\includegraphics[bb=0 0 216 173
height=2.4n,
width=3in
]{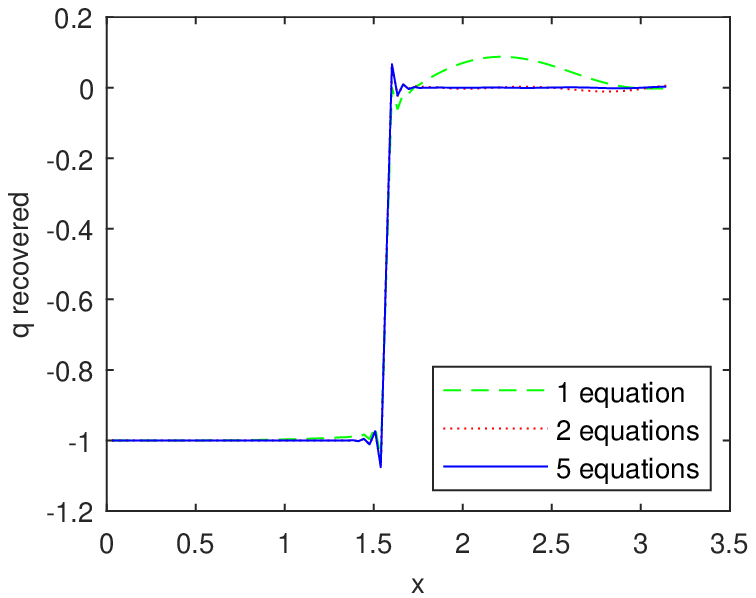}
\quad
\includegraphics[bb=0 0 216 173
height=2.4n,
width=3in
]{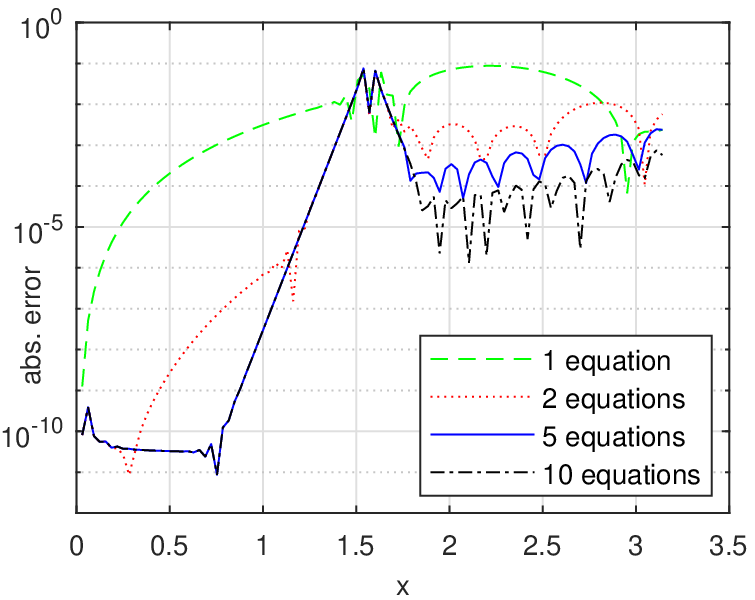}
\caption{On the left: the square well potential from Example
\ref{Ex 1} with $\ell=2$ recovered on the interval $(0,\pi]$ with $M=0$, $M=1$
and $M=4$ that corresponds to 1, 2 and 5 equations in the truncated system
(\ref{truncated system A}), respectively. On the right: absolute error of the
recovered potential for $M\in\{0, 1, 4, 9\}$ corresponding to 1, 2, 5 and 10
equations in the truncated system (\ref{truncated system A}), respectively.}
\label{Fig 1}
\end{figure}

\begin{figure}[tbh]
\centering
\includegraphics[bb=0 0 216 173
height=2.4n,
width=3in
]{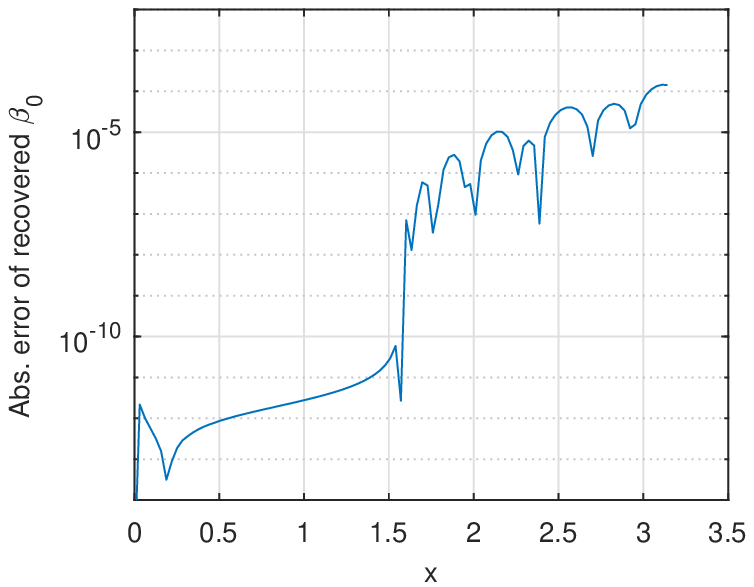}
\quad
\includegraphics[bb=0 0 216 173
height=2.4n,
width=3in
]{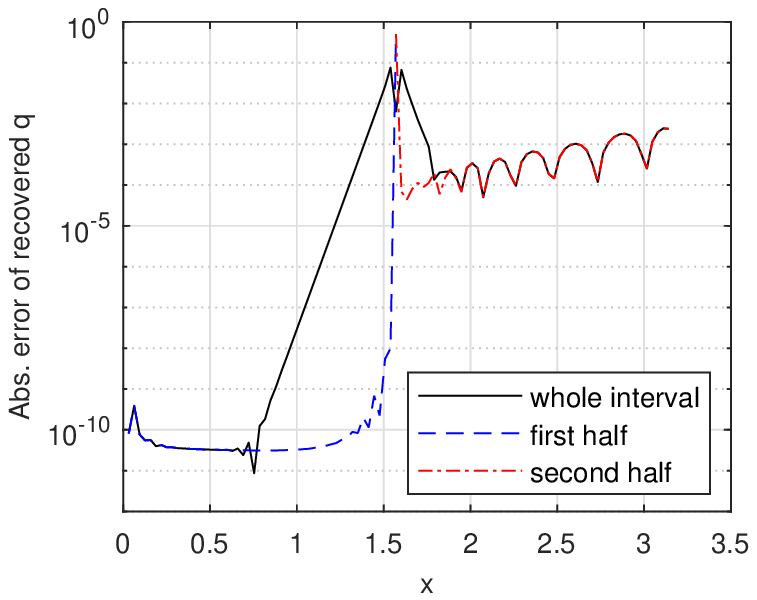}
\caption{On the left: absolute error of the recovered
coefficient $\beta_{0}$ for the square well potential from Example \ref{Ex 1}
with $\ell=2$, $M=4$ on the interval $(0,\pi]$. On the right: absolute error
of the recovered potential depending on the choise of the interval used for
spline interpolation of the coefficient $\beta_{0}$ and differentiation.}%
\label{Fig 2}%
\end{figure}

Note that the error increase in the recovered potential closer to the
discontinuity point $x=\pi/2$ is due to the error propagation in the spline
interpolation procedure. Indeed, on Figure \ref{Fig 2}, left plot, we show the
absolute error of the recovered coefficient $\beta_{0}$. As one can
appreciate, the error remains small almost up to the discontinuity point
$x=\pi/2$. So one can expect that applying numerical differentiation without
using values of $\beta_{0}$ from both sides of the discontinuity point, e.g.,
the finite difference or constructing a spline using the data from $[0,\pi/2]$
only, can reduce the error for values of $x$ close to $\pi/2$. Indeed, on
Figure \ref{Fig 2}, right plot, we show the error of the recovered potential
when the coefficient $\beta_{0}$ was approximated by a spline separately on
$[0,\pi/2]$ and on $[\pi/2,\pi]$. One can appreciate a higher accuracy close
to $x=\pi/2$.
\end{example}

\begin{example}
\label{Ex 2} The method gives excellent results for negative values of $\ell$
and for larger values of $\ell$ as well. Let us consider the same potential as
in Example \ref{Ex 1} but for $\ell=-1/2$ and $\ell=e^{3}$.

Note that for $\ell=-1/2$ the problem possesses an eigenvalue. We used the
method from \cite{KT2020improvedNeumann} to find its value, $\rho_{1}%
^{2}\approx-0.258265599397038$, with a corresponding norming constant
$c_{1}\approx0.469060824384319$.

On Figure \ref{Figure q ells} we show the absolute errors of the recovered potentials.

\begin{figure}[tbh]
\centering
\begin{tabular}
[c]{cc}%
$\ell=-1/2$ & $\ell=e^{3}$\\
\includegraphics[bb=0 0 216 173
height=2.4n,
width=3in
]{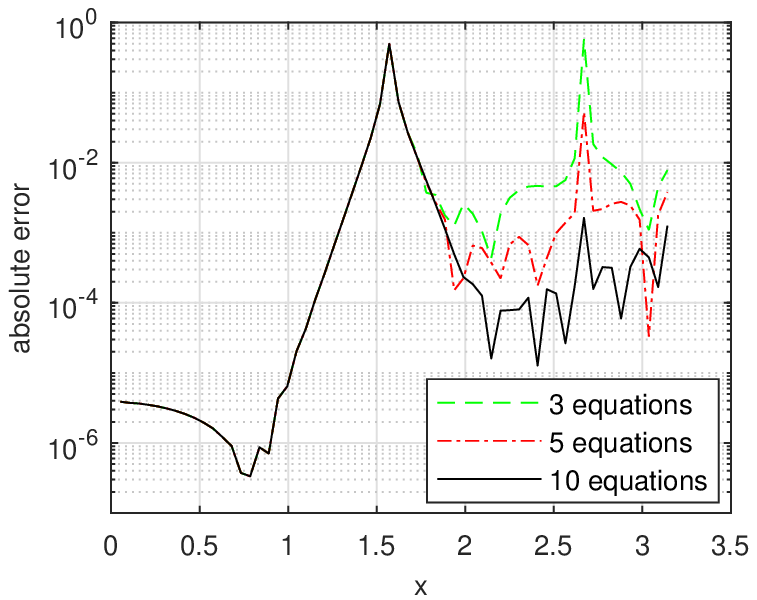} & 
\includegraphics[bb=0 0 216 173
height=2.4n,
width=3in
]{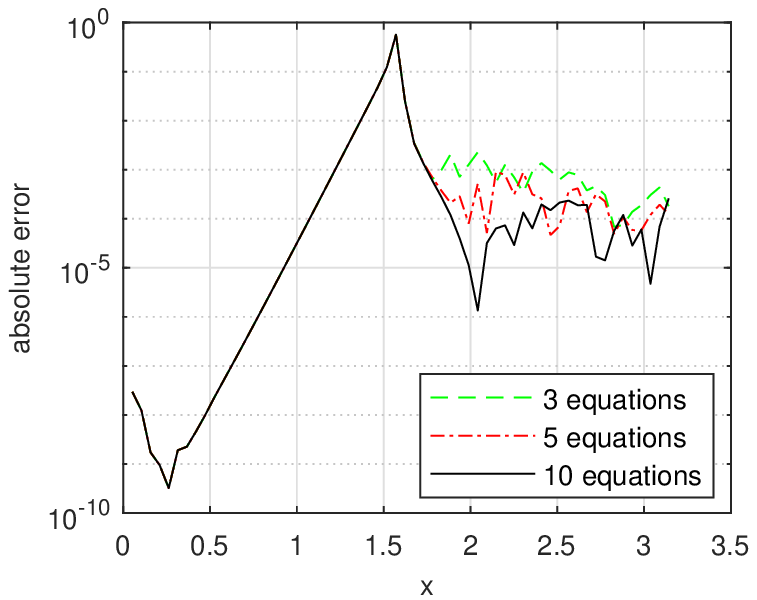}
\end{tabular}
\caption{Absolute errors of the recovered square well potential from Example
\ref{Ex 2} having $\ell=-1/2$ (on the left) and $\ell=e^{3}$ (on the right).
The potential was recovered on the interval $(0,\pi]$ and 3, 5 or 10 equations
were used in the truncated system \eqref{truncated system A}.}%
\label{Figure q ells}%
\end{figure}
\end{example}

\begin{example}
\label{Ex3}

Let us consider the equation with the Hulth\'{e}n effective potential
\begin{equation}
\label{eq HulthenEff}L_{1}u:=-u^{\prime\prime}+\left(  \ell(\ell+1) \left(
\frac{\delta}{1-e^{-\delta x}}\right)  ^{2}e^{-\delta x} -\frac{2\delta
e^{-\delta x}}{1-e^{-\delta x}}\right)  u=\rho^{2}u,\qquad x>0.
\end{equation}
Here $0<\delta<1$.

The Hulth\'{e}n potential $q_{H}(r)=-\frac{2\delta e^{-\delta x}}{1-e^{-\delta
x}}$ is known as a potential providing a better approximation to the screened
Coulomb (Yukawa) potential $q_{sc}(r)=-\frac{2e^{-\delta r}}{r}$ than the
ordinary Coulomb potential $-\frac{2}{r}$, see, e.g., \cite{Ma1954},
\cite{Greene}. However, it can be exactly solved only for zero angular
momentum, i.e., for $\ell=0$. Greene and Aldrich \cite{Greene} considered the
Hulth\'{e}n effective potential as an exactly solvable approximation for all
values of $\ell$. Equation \eqref{eq HulthenEff} can be transformed into the
form \eqref{perturbed Bessel equation} if one considers
\[
q(x)=\ell(\ell+1)\left(  \frac{\delta}{1-e^{-\delta x}}\right)  ^{2}e^{-\delta
x}-\frac{\ell(\ell+1)}{x^{2}}-\frac{2\delta e^{-\delta x}}{1-e^{-\delta x}}.
\]
Note that $q(x)\sim-\frac{2}{x}+\frac{\delta\ell(\ell+1)}{x}$ as
$x\rightarrow0$ and $q(x)\sim-\frac{\ell(\ell+1)}{x^{2}}$ as $x\rightarrow
\infty$, so the potential $q$ does not satisfy the condition
\eqref{Condition on q}. Nevertheless, the spectral problem for the original
equation \eqref{eq HulthenEff} possesses at most a finite number of negative
eigenvalues, see, e.g., \cite{Sohin1975}, so the corresponding quantum
scattering problem can be solved by the same method, see \cite{Chadan}.

To simplify the consideration below in what follows we assume that
$2\ell\not \in \mathbb{Z}$. Then the general solution of \eqref{eq HulthenEff}
has the form
\begin{equation}
\label{Hulthen gensol}%
\begin{split}
u(x)  &  =A y^{-\ell} e^{i\rho x} \, _{2}F_{1}\left(  -\ell-a_{\rho}-s_{\rho
},-\ell-a_{\rho}+s_{\rho};-2 \ell;y\right) \\
&  \quad+B y^{\ell+1}\,e^{i\rho x} \, _{2}F_{1}\left(  \ell+1-a_{\rho}%
-s_{\rho},\ell+1-a_{\rho}+s_{\rho};2 \ell+2;y\right)  ,
\end{split}
\end{equation}
where $y = 1-e^{-\delta x}$, $a_{\rho}= i\frac\rho\delta$ and $s_{\rho}=
\frac{\sqrt{2 \delta-\rho^{2} }}{\delta}$. This expression was obtained
solving transformed equation (7) from \cite{Greene} using Wolfram Mathematica 10.

Note that $y\rightarrow0$ as $x\rightarrow0$, hence the regular solution of
\eqref{eq HulthenEff} has the form
\begin{equation}
\varphi_{\ell}(\rho,x)=\frac{(1-e^{-\delta x})^{\ell+1}\,e^{i\rho x}}%
{\delta^{\ell+1}\left(  2\ell+1\right)  !!}\,_{2}F_{1}\left(  \ell+1-a_{\rho
}-s_{\rho},\ell+1-a_{\rho}+s_{\rho};2\ell+2;1-e^{-\delta x}\right)  .
\label{Hulthen regsol}%
\end{equation}
On the other hand, $y\rightarrow1$ as $x\rightarrow+\infty$, so the values of
the hypergeometric functions in \eqref{Hulthen gensol} are not defined by
their series expansions and to find the Jost solution we need to apply the
following analytic continuation \cite[(2.10.1)]{Erdelyi}
\[%
\begin{split}
\,_{2}F_{1}(a,b;c;z)  &  =\frac{\Gamma(c)\Gamma(c-a-b)}{\Gamma(c-a)\Gamma
(c-b)}\,_{2}F_{1}(a,b;a+b-c+1;1-z)\\
&  \quad+\frac{\Gamma(c)\Gamma(a+b-c)}{\Gamma(a)\Gamma(b)}(1-z)^{c-a-b}%
\,_{2}F_{1}(c-a,c-b;c-a-b+1;1-z).
\end{split}
\]
Then
\begin{align*}
u(x)  &  =A\frac{y^{-\ell}e^{i\rho x}\Gamma(-2\ell)\Gamma(2a_{\rho})}%
{\Gamma(-\ell+a_{\rho}+s_{\rho})\Gamma(-\ell+a_{\rho}-s_{\rho})}\,_{2}%
F_{1}(-\ell-a_{\rho}-s_{\rho},-\ell-a_{\rho}+s_{\rho};1-2a_{\rho};e^{-\delta
x})\\
\displaybreak[2]  &  +A\frac{y^{-\ell}e^{-i\rho x}\Gamma(-2\ell)\Gamma
(-2a_{\rho})}{\Gamma(-\ell-a_{\rho}-s_{\rho})\Gamma(-\ell-a_{\rho}+s_{\rho}%
)}\,_{2}F_{1}(-\ell+a_{\rho}+s_{\rho},-\ell+a_{\rho}-s_{\rho};1+2a_{\rho
};e^{-\delta x})\\
\displaybreak[2]  &  +B\frac{y^{\ell+1}e^{i\rho x}\Gamma(2\ell+2)\Gamma
(2a_{\rho})}{\Gamma(\ell+1+a_{\rho}+s_{\rho})\Gamma(\ell+1+a_{\rho}-s_{\rho}%
)}\,_{2}F_{1}(\ell+1-a_{\rho}-s_{\rho},\ell+1-a_{\rho}+s_{\rho};1-2a_{\rho
};e^{-\delta x})\\
&  +B\frac{y^{\ell+1}e^{-i\rho x}\Gamma(2\ell+2)\Gamma(-2a_{\rho})}%
{\Gamma(\ell+1-a_{\rho}-s_{\rho})\Gamma(\ell+1-a_{\rho}+s_{\rho})}\,_{2}%
F_{1}(\ell+1+a_{\rho}+s_{\rho},\ell+1+a_{\rho}-s_{\rho};1+2a_{\rho};e^{-\delta
x}).
\end{align*}
The first and the third terms behave like constant by $e^{i\rho x}$ when
$x\rightarrow\infty$, while the second and the forth terms behave like
constant by $e^{-i\rho x}$ when $x\rightarrow\infty$. Hence for the solution
$u$ to be the Jost solution, the coefficients $A$ and $B$ have to satisfy the
following system
\begin{align*}
&  A\cdot\frac{\Gamma(-2\ell)\Gamma(2a_{\rho})}{\Gamma(-\ell+a_{\rho}+s_{\rho
})\Gamma(-\ell+a_{\rho}-s_{\rho})}+B\cdot\frac{\Gamma(2\ell+2)\Gamma(2a_{\rho
})}{\Gamma(\ell+1+a_{\rho}+s_{\rho})\Gamma(\ell+1+a_{\rho}-s_{\rho})}%
=e^{i\pi\ell/2},\\
&  A\cdot\frac{\Gamma(-2\ell)\Gamma(-2a_{\rho})}{\Gamma(-\ell-a_{\rho}%
-s_{\rho})\Gamma(-\ell-a_{\rho}+s_{\rho})}+B\cdot\frac{\Gamma(2\ell
+2)\Gamma(-2a_{\rho})}{\Gamma(\ell+1-a_{\rho}-s_{\rho})\Gamma(\ell+1-a_{\rho
}+s_{\rho})}=0.
\end{align*}
Solving this system and using Euler's reflection formula $\Gamma
(1-z)\Gamma(z)=\frac{\pi}{\sin\pi z}$ we obtain that
\begin{align*}
A  &  =e^{i\pi\ell/2}\frac{\Gamma(1-2a_{\rho})\Gamma(2\ell+1)}{\Gamma
(\ell+1-a_{\rho}+s_{\rho})\Gamma(\ell+1-a_{\rho}-s_{\rho})},\\
B  &  =-e^{i\pi\ell/2}\frac{\Gamma(-2\ell)\Gamma(1-2a_{\rho})}{(2\ell
+1)\Gamma(-\ell-a_{\rho}+s_{\rho})\Gamma(-\ell-a_{\rho}-s_{\rho})}.
\end{align*}
The Jost function is given by
\[
F_{\ell}(\rho)=\lim_{x\rightarrow0}\frac{\left(  -\rho x\right)  ^{\ell}%
}{\left(  2\ell-1\right)  !!}\frac{Ae^{i\rho x}}{(1-e^{-\delta x})^{\ell}%
}=\frac{A}{(2\ell-1)!!}\left(  -\frac{\rho}{\delta}\right)  ^{\ell},
\]
hence
\begin{equation}
F_{\ell}(\rho)=\frac{e^{i\pi\ell/2}}{(2\ell-1)!!}\frac{\Gamma\left(
1-2i\frac{\rho}{\delta}\right)  \Gamma(2\ell+1)}{\Gamma\Bigl(\ell
+1-i\frac{\rho}{\delta}+\frac{\sqrt{2\delta-\rho^{2}}}{\delta}\Bigr)\Gamma
\Bigl(\ell+1-i\frac{\rho}{\delta}-\frac{\sqrt{2\delta-\rho^{2}}}{\delta
}\Bigr)}\left(  -\frac{\rho}{\delta}\right)  ^{\ell}. \label{Hulthen Fl}%
\end{equation}
Note that this expression is also well defined for values $\ell$ satisfying
$2\ell\in\mathbb{N}$.

The eigenvalues $\rho_{j}=i\tau_{j}$, $\tau_{j}\geq0$ correspond to zeros of
the function $F_{\ell}$. One can easily see that all such zeros coincide with
the values of $\rho$ for which $\Gamma\Bigl(\ell+1-i\frac{\rho}{\delta}%
-\frac{\sqrt{2\delta-\rho^{2}}}{\delta}\Bigr)=\infty$, i.e., when the argument
of the gamma function is a non-positive integer, which is equivalent to the
equation
\[
\ell+1+\frac{\tau}{\delta}-\frac{\sqrt{2\delta+\tau^{2}}}{\delta}=-m,\qquad
m\in\mathbb{N}_{0}.
\]
Squaring the equation we find that
\[
(\ell+1+m)^{2}+\frac{\tau^{2}}{\delta^{2}}+2\frac{\tau}{\delta}(\ell
+1+m)=\frac{2}{\delta}+\frac{\tau^{2}}{\delta^{2}},
\]
or
\[
\tau=\frac{1}{\ell+1+m}-\frac{\delta}{2}(\ell+1+m),\qquad m\in\mathbb{N}_{0}.
\]
Recalling that $\tau$ must be non-negative, we find that the set of
eigenvalues is given by
\[
\tau_{j}=\frac{1}{\ell+j}-\frac{\delta}{2}(\ell+j),\qquad j=1,\ldots,\left[
\sqrt{\frac{2}{\delta}}-\ell\right]  ,
\]
where $[\cdot]$ is the integer part function. The corresponding eigenfunctions
are given by
\[
\varphi_{\ell}(i\tau_{j},x)=\frac{(1-e^{-\delta x})^{\ell+1}\,e^{i\rho x}%
}{\delta^{\ell+1}\left(  2\ell+1\right)  !!}\,_{2}F_{1}\left(  -j+1,\ell
+1+\frac{2}{\delta(\ell+1)};2\ell+2;1-e^{-\delta x}\right)  .
\]
Note that the first argument of the hypergeometric function is a non-positive
integer, so the hypergeometric function reduces to a polynomial. The
corresponding norming constants can be easily obtained numerically.

On Figure \ref{Figure Coloumb} we show the recovered potential and the
absolute error. In this example the function $\left\vert F_{\ell}%
(\rho)\right\vert ^{-2}-1$ decays as $1/\rho$, and not as $1/\rho^{2}$ as was
considered in Subsection \ref{SubsectAsymptotics}. However, a similar
procedure was implemented to improve the computation of the integrals. The
interval $\rho\in[0, 1000]$ was used for numerical integration.

\begin{figure}[tbh]
\centering
\includegraphics[bb=0 0 216 173
height=2.4n,
width=3in
]{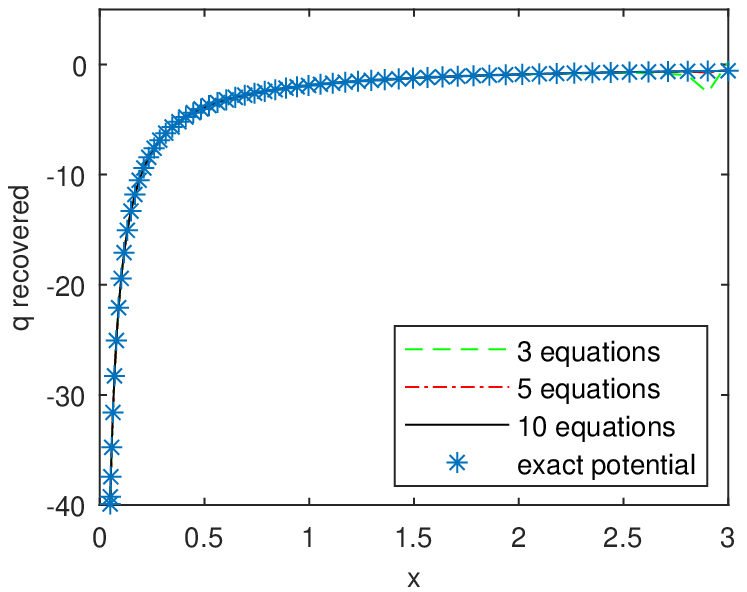} 
\quad
\includegraphics[bb=0 0 216 173
height=2.4n,
width=3in
]{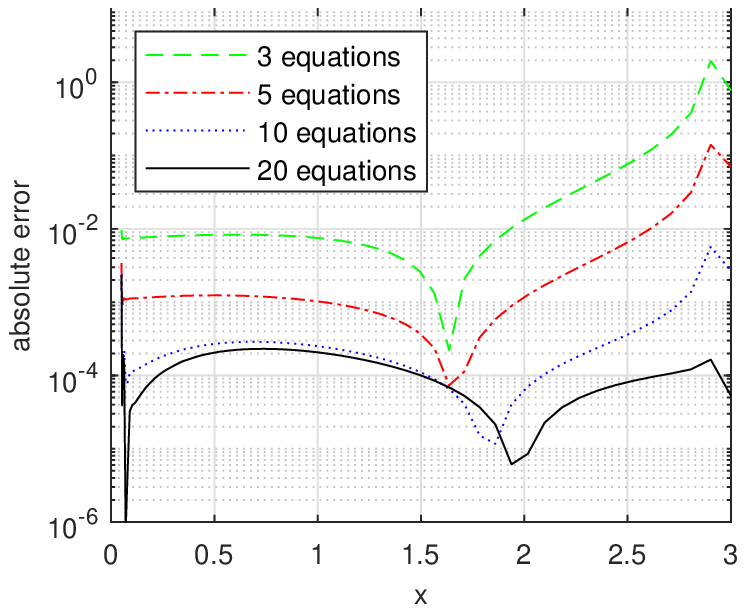}
\caption{ Recovered Hulth\'{e}n effective potential from
Example \ref{Ex3} with $\ell=1/3$ and $\delta=1/10$ (on the left) and the
absolute error (on the right). The potential was recovered on the interval
$[\frac{1}{20},3]$ using up to 20 equations in the truncated system
\eqref{truncated system A}.}%
\label{Figure Coloumb}%
\end{figure}

Proposition \ref{Prop CondNum} states that the condition numbers of the
matrices arising in the process remain bounded and that the method is stable
for a small noise. We illustrate these statements by Figure
\ref{Figure Stability}. On the left we show the smallest and the largest
eigenvalues of the coefficient matrix of the truncated system
\eqref{system Kantorovich} as the function of the truncation parameter $M$. As
one can see, the condition number, which is equal to the quotient of these
eigenvalues, remains bounded independently of the number of equations used. On
the right we show the potential recovered from the noisy data $\{\tau
_{j},c_{j}\}_{j=1}^{4}$ and $F_{\ell}(\rho)$, $\rho\in\left\{  \frac{k}%
{10}\right\}  _{k=0}^{1000}$, 10\% noise was added to all the values.

\begin{figure}[tbh]
\centering
\includegraphics[bb=0 0 216 173
height=2.4n,
width=3in
]{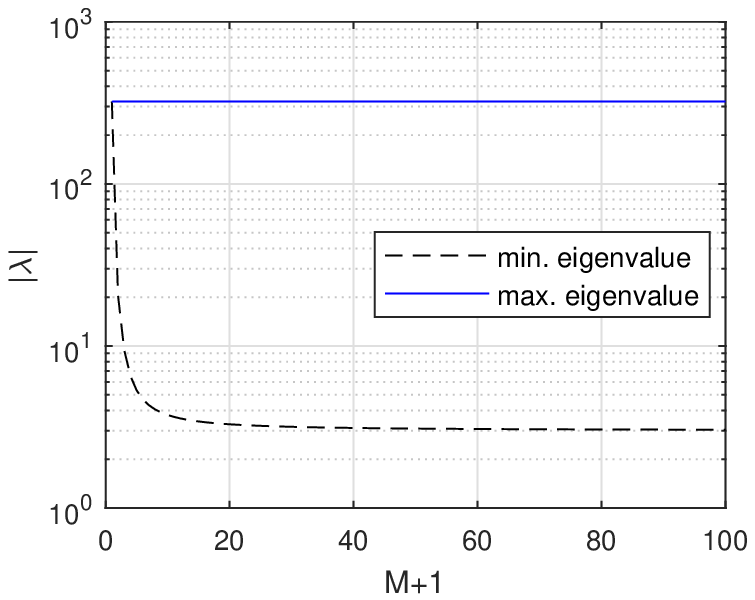} 
\quad
\includegraphics[bb=0 0 216 173
height=2.4n,
width=3in
]{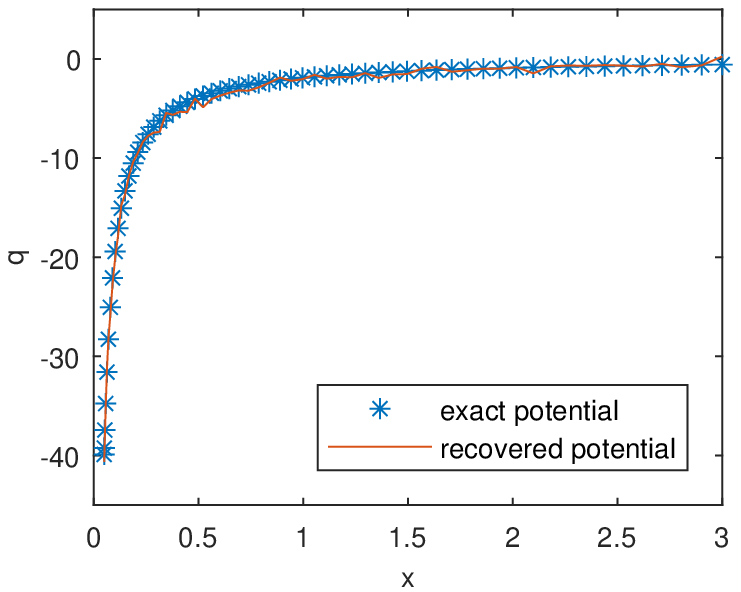}
\caption{ The Hulth\'{e}n effective potential from Example \ref{Ex3}
having $\ell=1/3$ and $\delta=1/10$ is considered. On the left: the smallest
and the largest eigenvalues of the coefficient matrix of the truncated system
for $x=3$ as the function of $M$. On the right: potential recovered from the
noisy data.}%
\label{Figure Stability}%
\end{figure}
\end{example}

\section{Conclusions}

A direct and simple method for solving the inverse quantum scattering problem
for an arbitrary angular momentum $\ell\geq-1/2$ is presented. Numerical
solution of the problem reduces to the solution of a system of linear
algebraic equations from which the first component of the solution vector is
sufficient for recovering the potential. Numerical results reveal a remarkable
accuracy, stability and a fast convergence of the method.

\appendix

\section{On the asymptotic behaviour of the function $F_{\ell}$}

\label{SectAppendix} According to \cite[Lemma B.5]{Teschl} the function
$F_{\ell}$ admits the following integral representation
\begin{equation}
\label{Fl integral}F_{\ell}(\rho) = 1+\int_{0}^{\infty}\psi^{0}_{\ell}(\rho,x)
\varphi_{\ell}(\rho,x) q(x) \,dx,
\end{equation}
where $\varphi_{\ell}$ is the regular solution of
\eqref{perturbed Bessel equation} considered in Section \ref{Sect2} and
$\psi_{\ell}^{0}$ denotes the Jost solution of
\eqref{perturbed Bessel equation} with $q\equiv0$. We also denote the regular
solution for $q\equiv0$ by $\varphi_{\ell}^{0}$.

First, we recall some estimates from \cite{Teschl} and \cite{Teschl2016}. The
solutions $\varphi_{\ell}^{0}$ and $\psi_{\ell}^{0}$ of the unperturbed
equation are given by
\begin{align}
\varphi_{\ell}^{0}(\rho,x)  & = \rho^{-\ell-1/2} \sqrt{\frac{\pi x}{2}}
J_{\ell+1/2} (\rho x),\label{sol phi0}\\
\psi_{\ell}^{0}(\rho,x)  & = i\rho^{\ell+1/2} \sqrt{\frac{\pi x}{2}}
H_{\ell+1/2}^{(1)}(\rho x).\label{sol psi0}%
\end{align}
Here $H_{\ell+1/2}^{0}$ is the Hankel function of the first kind. The
following estimates hold for $\rho\in\mathbb{R}$, $\ell>-1/2$
\begin{equation}
|\varphi_{\ell}^{0}(\rho,x)|\le C\left( \frac{x}{1+|\rho|x}\right) ^{\ell
+1},\qquad|\psi_{\ell}^{0}(\rho,x)|\le C\left( \frac{x}{1+|\rho|x}\right)
^{-\ell}.\label{phi0 psi0 est}%
\end{equation}
For $\ell=-1/2$ the first estimate remains valid, and the second changes to
\begin{equation}
\label{psi0m12}|\psi_{-1/2}^{0}(\rho,x)|\le C\left( \frac{x}{1+|\rho|x}\right)
^{1/2}\left( 1-\log\frac{|\rho|x}{1+|\rho|x}\right) .
\end{equation}

Considering the difference between the regular solutions $\varphi_{\ell}$ and
$\varphi_{\ell}^{0}$, the following estimate follows from \cite[(2.21) and
(2.23)]{KosSakhTesh2010} and \cite[Lemma B.2]{Teschl}
\begin{equation}
\label{phi phi0}%
\begin{split}
|\varphi_{\ell}(\rho,x) - \varphi_{\ell}^{0}(\rho,x)|  &  \le\sum
_{n=1}^{\infty}\frac{C^{n+1}}{n!} \left( \frac{x}{1+|\rho|x}\right)  ^{\ell+1}
e^{|\operatorname{Im}\rho|x} \biggl(\int_{0}^{x} \frac{y|\bar q(y)|}%
{1+|\rho|y}dy\biggr)^{n}\\
& = C\left( \frac{x}{1+|\rho|x}\right)  ^{\ell+1} e^{|\operatorname{Im}\rho|x}
\left(  \exp\biggl(C\int_{0}^{x} \frac{y|\bar q(y)|}{1+|\rho|y}%
dy\biggr)-1\right) ,
\end{split}
\end{equation}
here $\bar q(x) = q(x)$ if $\ell>-1/2$ and $\bar q(x) = \bigl(1-\log(\frac
x{1+x})\bigr)q(x)$ if $\ell=-1/2$.

For large values of arguments the functions $\varphi_{\ell}^{0}$ and
$\psi_{\ell}^{0}$ can be approximated using the asymptotic formulas for the
functions $J_{\nu}$ and $H_{\nu}^{(1)}$. We have (see \cite[(9.2.5)--(9.2.10)]%
{AbramowitzStegunSpF} for $z\in\mathbb{R}$
\begin{equation}
\label{J and H series}J_{\nu}(z) = \sqrt{\frac{2}{\pi z}} \bigl(P(\nu
,z)\cos\chi- Q(\nu,z)\sin\chi\bigr),\qquad H_{\nu}^{(1)}(z) = \sqrt{\frac
{2}{\pi z}} \bigl(P(\nu,z)+i Q(\nu,z)\bigr)e^{i\chi},
\end{equation}
where $\chi= z - (\frac{\nu}2+\frac14)\pi$ and
\begin{align}
P(\nu,z)  & = \sum_{k=0}^{[\nu/2+3/4]} (-1)^{k}\frac{(\nu,2k)}{(2z)^{2k}} +
\theta_{1}(z) \frac{(\nu, 2[\nu/2+7/4])}{(2z)^{2[\nu/2+7/4]}},\label{Pas}\\
Q(\nu,z)  & = \sum_{k=0}^{[\nu/2+1/4]} (-1)^{k}\frac{(\nu,2k+1)}{(2z)^{2k+1}}
+ \theta_{2}(z) \frac{(\nu, 2[\nu/2+5/4]+1)}{(2z)^{2[\nu/2+5/4]+1}%
},\label{Qas}%
\end{align}
with $|\theta_{1,2}|\le1$. From \eqref{J and H series}--\eqref{Qas} it follows
that for all $z\ge1$
\begin{align*}
J_{\nu}(z)  &  = \sqrt{\frac{2}{\pi z}} \left(  \cos\left( z-\frac{\nu\pi}2 -
\frac{\pi}4\right)  + \frac{4\nu^{2}-1}{8z}\sin\left( z-\frac{\nu\pi}2 -
\frac{\pi}4\right)  + O\left( \frac1{z^{2}}\right) \right) ,\\
H_{\nu}^{(1)}(z)  & = \sqrt{\frac{2}{\pi z}}e^{i\left( z-\frac{\nu\pi}2 -
\frac{\pi}4\right) }\left( 1+ i\frac{4\nu^{2}-1}{8z}+ O\left( \frac1{z^{2}%
}\right) \right) ,
\end{align*}
here $O(1/z^{2})$ means that there exists a constant $C$ such that the
remainder is bounded by $C/z^{2}$ for all $z\ge1$. Taking the product and
expanding $\cos\chi$ and $\sin\chi$ via the sum and difference of $e^{i\chi}$
and $e^{-i\chi}$, we obtain that
\begin{equation}
\label{ProdAsympt}J_{\nu}(z) H_{\nu}^{(1)}(z) = \frac{1}{\pi z} \left(  1 +
e^{2i\left( z-\frac{\nu\pi}2 - \frac{\pi}4\right) } + i\frac{4\nu^{2}-1}%
{4z}e^{2i\left( z-\frac{\nu\pi}2 - \frac{\pi}4\right) } + O\left( \frac
1{z^{2}}\right) \right)
\end{equation}
for all $z\ge1$.

\begin{remark}
One can easily deduce from \eqref{sol phi0}, \eqref{sol psi0} and
\eqref{ProdAsympt} that
\[
2\rho\varphi_{\ell}^{0}(\rho,x) \psi_{\ell}^{0}(\rho,x) = i\left(
1+e^{2i\left( \rho x-\frac{(\ell+1) \pi}2\right) }+O\left( \frac1\rho\right)
\right) ,
\]
and does not converge when $\rho\to\infty$ contrary to what is stated in
\cite[Remark 2.14]{Teschl2016}.
\end{remark}

We refer the reader to \cite{Mendoza2016} for the definition of the total
variation of a function (denoted by $V(f;\mathbb{R})$) and the class of
bounded variation vanishing at infinity functions (denoted by $BV_{0}%
(\mathbb{R})$). We need the following two properties of the functions from
$BV_{0}(\mathbb{R})$ class.

\begin{lemma}
\label{Lemma BV prod} Let $f\in BV_{0}(\mathbb{R})$ and $g\in BV(\mathbb{R})$.
Then $f\cdot g\in BV_{0}(\mathbb{R})$, and $V(f\cdot g;\mathbb{R})\le
V(f;\mathbb{R})\cdot\bigl(V(g;\mathbb{R})+\|g\|_{L_{\infty}(\mathbb{R}%
)}\bigr)$.
\end{lemma}

\begin{proof}
It is well-known that the product of two functions of bounded variation is again a function of bounded variation, see, e.g. \cite{GC2002}, moreover
\[
V(f\cdot g;\mathbb{R})\le  V(g;\mathbb{R})\cdot \sup_{\mathbb{R}}|f| +  V(f;\mathbb{R})\cdot \sup_{\mathbb{R}}|g|.
\]
Now the statement follows observing that for $BV_0$ functions one has $\sup\limits_{\mathbb{R}}|f|\le V(f;\mathbb{R})$.
\end{proof}

\begin{lemma}
[{\cite[Corollary 10]{Mendoza2016}}]\label{Lemma Fourier decay} If $f\in
BV_{0}(\mathbb{R})$, then for all $\omega\in\mathbb{R}\setminus\{0\}$ its
Fourier transform $\hat f$ is defined and satisfies
\[
|\hat f(\omega)|\le\frac{V(f;\mathbb{R})}{|\omega|}.
\]

\end{lemma}

Now we can formulate the main result of this section.

\begin{proposition}
Suppose that the potential $q\in L_{1}(0,\infty)\cap BV_{0}[0,\infty)$. Then
the asymptotics \eqref{F asymptotics} holds.
\end{proposition}

\begin{proof}
Due to the property $F_{\ell}(-\rho)=\overline{F_{\ell}}(\rho)$ we may assume that $\rho>0$.
First we assume that $\ell>-1/2$.
Let us rewrite \eqref{Fl integral} as
\[
\begin{split}
F_\ell(\rho) &= 1+\int_0^{1/\rho} \psi_\ell^0 (\rho, x) \varphi_\ell^0 (\rho,x)q(x)\,dx + \int_{1/\rho}^\infty \psi_\ell^0 (\rho, x) \varphi_\ell^0 (\rho,x)q(x)\,dx \\
&\quad + \int_0^\infty \psi_\ell^0 (\rho, x)\bigl(\varphi_\ell(\rho,x) -\varphi_\ell^0 (\rho,x)\bigr)q(x)\,dx =: 1+I_1 + I_2 + I_3.
\end{split}
\]
The integral $I_1$ can be estimated using \eqref{phi0 psi0 est} and noting that functions of bounded variation are bounded,
\[
|I_1|\le C^2\int_0^{1/\rho} \frac{x|q(x)|}{1+\rho x}\,dx \le \frac{C^2}{\rho}\int_0^{1/\rho} |q(x)|\,dx \le \frac{C_1}{\rho^2}.
\]
To estimate the integral $I_2$ we utilize \eqref{sol phi0}, \eqref{sol psi0} and \eqref{ProdAsympt} and obtain
\begin{align*}
I_2 & = \frac{i}{2\rho} \int_{1/\rho}^\infty q(x)\,dx + \frac{ie^{-i(\ell+1)\pi}}{2\rho} \int_{1/\rho}^\infty q(x) e^{2i\rho x}\,dx \\
& \quad -\frac{\ell(\ell+1)e^{-i(\ell+1)\pi}}{2\rho^2} \int_{1/\rho}^\infty \frac{q(x)}{x} e^{2i\rho x}\,dx + \frac{i}{2\rho} \int_{1/\rho}^\infty q(x)O\left(\frac 1{(\rho x)^2}\right)\,dx\\
& = I_4+I_5+I_6+I_7.
\end{align*}
Now we have
\[
I_4 = \frac{i}{2\rho} \int_{0}^\infty q(x)\,dx - \frac{i}{2\rho} \int_0^{1/\rho} q(x)\,dx = \frac{i}{2\rho} \int_{0}^\infty q(x)\,dx + O\left(\frac 1{\rho^2}\right),
\]
where we used that $q$ is bounded.
To estimate the integral in $I_5$ note that it can be considered as $\int_{-\infty}^\infty g(x) e^{2i\rho x}\,dx$, where $g(x) = \tilde q(x) \cdot \mathbf{1}_{[1/\rho,\infty)}(x)$, $\tilde q$ is an extension of $q$ to $\mathbb{R}$ by zero and $\mathbf{1}_{A}$ is the characteristic function of the set $A$. Both functions $\tilde q$ and $\mathbf{1}_{[1/\rho,\infty)}$ are of bounded variation and one can easily see from Lemma \ref{Lemma BV prod} that $V(g;\mathbb{R})\le 4V(q;[0,\infty))$. Now applying Lemma \ref{Lemma Fourier decay} we obtain that
\[
|I_5| = \frac{1}{2\rho}\left|\int_{-\infty}^\infty g(x) e^{2i\rho x}\,dx\right| \le \frac{4V(q;[0,\infty))}{4\rho^2}.
\]
The estimate for the integral $I_6$ is similar noting that the function $\frac 1x \mathbf{1}_{[1/\rho,\infty)}(x)\in BV_0(\mathbb{R})$ and $V(\frac 1x \mathbf{1}_{[1/\rho,\infty)}(x);\mathbb{R})\le 2\rho$.
Recalling the meaning of the $O$ symbol, we have for $I_7$
\[
|I_7|\le \frac{C}{2\rho^3}\int_{1/\rho}^\infty \frac{|q(x)|}{x^2}\,dx \le \frac{C}{2\rho^3}\int_{1/\rho}^\infty \frac{C_2}{x^2}\,dx = \frac{CC_2}{2\rho^2}.
\]
Finally, for the integral $I_3$ we utilize \eqref{phi0 psi0 est} and \eqref{phi phi0} and obtain
\[
|I_3|\le C^2\int_0^\infty \frac{x|q(x)|}{1+\rho x} \left( \exp\biggl(C\int_0^x \frac{y|q(y)|}{1+\rho y}dy\biggr)-1\right) \,dx.
\]
Since $\int_0^x \frac{y|q(y)|}{1+|\rho|y}dy\le \frac{1}{\rho}\int_0^x |q(x)|\,dx\le \frac 1{\rho} \|q\|_{L_1(0,\infty)}$, we have
\[
\left|\exp\biggl(C\int_0^x \frac{y|q(y)|}{1+\rho y}dy\biggr)-1\right| \le \exp\left(\frac{C\|q\|_{L_1(0,\infty)}}{\rho}\right)-1 = O\left(\frac 1\rho\right),
\]
hence $I_3 = O(\frac 1{\rho^2})$. Combining all the estimates we obtain the statement.
Now assume that $\ell=-1/2$. From all the integrals $I_1,\ldots,I_7$ only the integrals $I_1$ and $I_3$ have to be treated differently from the case $\ell>-1/2$. For the first we use the estimate \eqref{psi0m12} and boundedness of $q$ and obtain
\[
|I_1|\le C^2\int_0^{1/\rho} \frac{x|q(x)|}{1+\rho x}\left(1-\log\frac{|\rho|x}{1+|\rho|x}\right)\,dx \le
\frac{C^2C_q}{\rho}\left(\frac 1\rho - \int_0^{1/\rho} \log\frac{|\rho|x}{1+|\rho|x} \,dx \right)= \frac{C_3}{\rho^2}(1+\log 4).
\]
For the second integral note that since the function $q$ is bounded, the function $\bar q\in L_1(0,\infty)$, hence the same proof with the replacement of $q$ by $\bar q$ works.
\end{proof}

\end{document}